\documentclass[11pt]{article}

\usepackage{dsfont}

\def\colorful{0}

\usepackage{amsthm,amsfonts,amsmath,amssymb,epsfig,color,float,graphicx,verbatim, enumitem}
\usepackage{multirow}
\usepackage{algorithm}
\usepackage[noend]{algpseudocode}

\newif\ifhyper\IfFileExists{hyperref.sty}{\hypertrue}{\hyperfalse}
\hypertrue
\ifhyper\usepackage{hyperref}\fi

\usepackage{enumitem}

\usepackage{framed}
\usepackage{nicefrac}

\def\nnewcolor{1}
\ifnum\nnewcolor=1

\fi
\ifnum\nnewcolor=0

\fi

\ifnum\colorful=1
\newcommand{\new}[1]{{\color{red} #1}}

\newcommand{\newblue}[1]{{\color{blue} #1}}
\else
\newcommand{\new}[1]{{#1}}

\newcommand{\newblue}[1]{{#1}}

\fi

\newtheorem{theorem}{Theorem}[section]

\newtheorem{lemma}[theorem]{Lemma}
\newtheorem{informal theorem}[theorem]{Theorem (informal statement)}

\newtheorem{assm}[theorem]{Assumption}

\newtheorem{remark}[theorem]{Remark}

\theoremstyle{definition}
\newtheorem{definition}[theorem]{Definition}

\newcommand{\p}{\mathbf{P}}

\newcommand{\R}{\mathbb{R}}

\newcommand{\Z}{\mathbb{Z}}

\newcommand{\eps}{\epsilon}

\renewcommand{\Pr}{\mathbf{Pr}}
\newcommand{\poly}{\mathrm{poly}}

\title{Non-Gaussian Component Analysis via Lattice Basis Reduction}

\author{
Ilias Diakonikolas\thanks{Supported by NSF Medium Award CCF-2107079,
NSF Award CCF-1652862 (CAREER), a Sloan Research Fellowship, and
a DARPA Learning with Less Labels (LwLL) grant. Some of this work was performed while the author
was visiting the Simons Institute for the Theory of Computing.}\\
University of Wisconsin-Madison\\
{\tt ilias@cs.wisc.edu}\\
\and
Daniel M. Kane\thanks{Supported by NSF Medium Award CCF-2107547,
NSF Award CCF-1553288 (CAREER), and a Sloan Research Fellowship.}\\
University of California, San Diego\\
{\tt dakane@cs.ucsd.edu}\\
}

\begin{document}

\maketitle

\begin{abstract}
Non-Gaussian Component Analysis (NGCA) is the following distribution learning problem:
Given i.i.d.\ samples from a distribution on $\R^d$ that is non-gaussian in a hidden direction
$v$ and an independent standard Gaussian in the orthogonal directions, the goal is to approximate
the hidden direction $v$. Prior work~\cite{DKS17-sq} provided formal evidence 
for the existence of an information-computation tradeoff for NGCA 
under appropriate moment-matching conditions on the univariate non-gaussian distribution $A$. 
The latter result does not apply when the distribution $A$ is discrete.
A natural question is whether information-computation tradeoffs persist in this setting. 
In this paper, we answer this question in the negative 
by obtaining a sample and computationally efficient algorithm for NGCA 
in the regime that $A$ is discrete or nearly discrete, in a well-defined technical sense. 
The key tool leveraged in our algorithm is the LLL method~\cite{LLL82}
for lattice basis reduction. 
\end{abstract}

\setcounter{page}{0}

\thispagestyle{empty}

\newpage

\section{Introduction} \label{sec:intro}

\subsection{Background and Motivation} \label{ssec:background}

\paragraph{Non-Gaussian Component Analysis.}
Non-gaussian component analysis (NGCA) is a distribution learning problem
modeling the natural task of finding ``interesting'' directions in high-dimensional data.
As the name suggests, the objective is to find a ``non-gaussian'' direction (or, more generally, low-dimensional
subspace) in a high-dimensional dataset, under a natural generative model.
NGCA was defined in~\cite{JMLR:blanchard06a} and subsequently studied
from an algorithmic standpoint in a number of works,
see, e.g.,~\cite{VX11, TanV18, GoyalS19} and references therein.

For concreteness, we start by defining the relevant family of high-dimensional distributions.

\begin{definition} [High-Dimensional Hidden Direction Distribution] \label{def:pv-hidden}
{For a distribution $A$ on the real line with probability density function $A(x)$ and}
 a unit vector $v \in \R^d$, consider the distribution over $\R^d$ with probability density function
$\p^{A}_v(x) = A(v \cdot x) \exp\left(-\|x - (v \cdot x) v\|_2^2/2\right)/(2\pi)^{(d-1)/2}.$
That is, $\p^{A}_v$ is the product distribution whose orthogonal projection
onto the direction of $v$ is $A$,
and onto the subspace perpendicular to $v$
is the standard $(d-1)$-dimensional normal distribution.
\end{definition}

The NGCA learning problem is the following: Given i.i.d.\ samples from a distribution
$\p^{A}_v$ on $\R^d$, where the direction $v$ is unknown, find (or approximate) $v$. The standard
formulation assumes that the univariate distribution $A$ is known to the algorithm,
it matches its first $k$ moments with $N(0, 1)$, for some $k \in \Z_+$,
and there is a non-trivial difference in the moment of order $(k+1)$.

\paragraph{Information-Computation Tradeoffs for NGCA.}

\new{Since $A$ has its $(k+1)^{th}$ moment differing from that of a standard Gaussian,
a moment computation on $\p^{A}_v$ allows us to approximate $v$ in roughly $O(d^{k+1})$ samples and time.}
Interestingly, ignoring computational considerations, the NGCA problem can usually be solved with $O(d)$ samples.
Perhaps surprisingly, the aforementioned simple method \new{(requiring $\Omega(d^{k+1})$ samples)}
is qualitatively the best known \new{sample-polynomial time} algorithm for the problem.
\new{Given this state of affairs, it is natural to ask} 
whether this information-computation gap is inherent for the problem itself.

In prior work,~\cite{DKS17-sq} provided formal evidence for the existence of an {\em information-computation
tradeoff} for NGCA under appropriate assumptions on the \new{univariate non-gaussian} distribution $A$.
The~\cite{DKS17-sq} result holds for a restricted model of computation,
known as the Statistical Query (SQ) model.
Statistical Query (SQ) algorithms are the class of algorithms
that are only allowed to query expectations of bounded functions
of the underlying distribution rather than directly access samples.
The SQ model was introduced by Kearns~\cite{Kearns:98}
and has been extensively studied in learning theory.
A recent line of work, see, e.g.,~\cite{FeldmanGRVX17, FeldmanPV15, FeldmanGV17},
generalized the SQ framework for search problems over distributions.
The reader is referred to~\cite{Feldman16b} for a survey.

In more detail, the SQ lower bound of~\cite{DKS17-sq}
applies even for the (easier) hypothesis testing version of NGCA, where the goal is to distinguish
between the standard Gaussian $N(0, I)$ on $\R^d$ and a planted distribution $\p_v^A$, for a hidden direction $v$.
(Hardness for hypothesis testing can easily be used to derive hardness for the corresponding search problem.)
Roughly speaking, they established the following generic SQ-hardness result:

\begin{quote}
{\bf Informal Theorem~\cite{DKS17-sq}:}
Let $A$ be a one-dimensional distribution that matches its first
$k$ moments with the standard Gaussian $G = N(0,1)$
and its chi-squared norm with $G$, $\chi^2(A,G)$, is \new{finite}.
Suppose we want to distinguish between $N(0, I)$ on $\R^d$ and the
distribution $\p^A_v$ for a random direction $v$.
Then any SQ algorithm for this testing task requires either
at least $d^{\Omega(k)} /\chi^2(A,G)$ samples or at least $2^{d^{\Omega(1)}}$ time.
\end{quote}

\new{A concrete application of the above result, given in~\cite{DKS17-sq}, is
an SQ lower bound for the classical problem of learning mixtures of high-dimensional Gaussians.
To obtain the hard family of instances, we take the one-dimensional distribution $A$
be a mixture of univariate Gaussians $\sum_{i=1}^k w_i N(\mu_i, \sigma^2)$
with pairwise separated and bounded means $\mu_i$
and common variance $\sigma^2 = 1/\poly(k)$
such that $A$ matches $\Omega(k)$ moments with $N(0, 1)$.
Moreover, $A$ will have total total variation distance at least $1/2$ from $N(0, 1)$.
Then, each distribution $\p_v^A$ will look like a collection of $k$ ``parallel pancakes'',
in which the means lie on a line (corresponding to the smallest eigenvalue
of the identical covariance matrices of the components). The orthogonal directions
will have an eigenvalue of one, which is much larger than the smallest eigenvalue.
}

More broadly, the aforementioned generic SQ lower bound~\cite{DKS17-sq}
has been the basis for a host of new and near-optimal 
information-computation tradeoffs (in the SQ model)
for high-dimensional estimation tasks, including 
robust mean and covariance estimation~\cite{DKS17-sq},
robust sparse mean estimation~\cite{DKS17-sq},
adversarially robust learning~\cite{BPR18},
robust linear regression~\cite{DKS19}, list-decodable estimation~\cite{DKS18-list, DKPPS21},
learning simple neural networks~\cite{DiakonikolasKKZ20}, and robust supervised learning
in a variety of noise models~\cite{DKZ20, DiakonikolasKPZ21, DK20-Massart-hard, DKKTZ21-benign}.
Interestingly, subsequent work has obtained additional evidence of hardness for some of these problems
via reductions from lattice problems~\cite{BrunaRST21} and variants of the planted clique problem~\cite{BrennanB20}.

\paragraph{Motivation for This Work.}
\new{
Interestingly, the generic SQ lower bound of~\cite{DKS17-sq}
is vacuous for the natural setting
where the distribution $A$ is discrete (in which case, we have $\chi^2(A,N(0,1)) = \infty$) or,
more generally, when $A$ has very large chi-squared norm with the standard Gaussian.
\newblue{More specifically, for the parallel pancakes distribution described above, one needs the ``thickness
parameter'' (corresponding to the eigenvalue of the covariance in the hidden direction)
to be at least inverse exponential in the dimension.}
A natural question,
which served as one of the motivations for this work, is whether  information-computation tradeoffs
persist for the discrete case.

Consider for example the case where $A$ is supported on a discrete domain of size $k$
and matches its first $\Omega(k)$ moments with $N(0, 1)$. This corresponds to the special
case of the parallel pancakes distribution, where the component covariances are degenerate --- 
having zero eigenvalue in the hidden-direction.
Does any efficient algorithm for these instances require $d^{\omega_k(1)}$ samples?

We answer these questions in the negative
by designing a sample and computationally efficient algorithm
for NGCA when $A$ is discrete or nearly discrete
in a well-defined sense (Assumption~\ref{cond:A-disc}).
The key tool leveraged in \new{our reuslt} is the LLL \new{algorithm}~\cite{LLL82}
for lattice basis reduction. We note that prior work~\cite{BrunaRST21, SZB21}
had used the LLL algorithm to obtain efficient learners for related problems
that could be viewed as special cases of NGCA.

\paragraph{Connection with Sum-of-Squares (SoS) and Low-Degree Tests.}
Before we proceed with a detailed description of our results, a final remark is in order.
As already mentioned, the SQ lower bounds of~\cite{DKS17-sq} are vacuous when $A$ is a discrete distribution.
On the other hand, recent work has established information-computation tradeoffs for NGCA when 
$A$ is supported on $\{-1, 0, 1\}$, both for low-degree polynomial tests~\cite{MW21} and for
SoS algorithms~\cite{GhoshJJPR20}. At first sight, these hardness results combined with 
our algorithm might appear to cast doubt
on the validity of the low-degree conjecture~\cite{Hopkins-thesis}.
We note, however, that the latter conjecture only posits that a {\em noisy} version of the corresponding problem
is computationally hard (as opposed to the problem itself) --- a statement that appears to hold true in our setting.
Conceptually, we view our algorithmic contribution as a novel example of an efficient algorithm
(beyond Gaussian elimination) not captured by the aforementioned restricted models of computation.
}

\subsection{Our Contributions} \label{ssec:results}

We consider the NGCA learning problem under the following assumption:

\begin{assm} \label{cond:A-disc}
The distribution $A$ on $\R$ is such that:
\begin{enumerate}
\item\label{LatticeCondition} There exist $r_j \in \R$ for $j \in [k]$ with $|r_j|  = O(1)$,
$B \in \Z_+$, and $\eps>0$ such that
a sample $y \sim A$ is deterministically within additive $\eps$ of some number
of the form $\sum_{j=1}^k n_j r_j$, for $n_j \in \Z$ with $|n_j|\leq B$ for all $j \in [k]$.

\item \label{anticoncentrationCondition}
The distribution $A$ is anti-concentrated around $0$,
specifically $\Pr_{X \sim A}[|X|  > 1/d] > 1/d$.

\item \label{concentrationCondition}
The distribution $A$ is concentrated around $0$,
specifically $\Pr_{X \sim A}[|X| > \poly(d)] < 1/d$.
\end{enumerate}
\end{assm}

\new{Some comments are in order to interpret Assumption~\ref{cond:A-disc}.}
Condition~\ref{LatticeCondition} above is the critical condition requiring
that $A$ is approximately supported on points 
that are (small) integer linear combinations of the $r_j$'s.
\new{This is the key condition that underlies our main technique.
Notice that this condition can be satisfied by any distribution $A$ that has support size at most $k$,
or even a distribution $A$ that is supported on $k$ intervals, each of length at most $\eps$. 
In fact, it is sufficient for $A$ to be $O(1/d)$-close in total variation distance to such a distribution, 
as there will be a constant probability that any $O(d)$ sample set drawn from it are supported on the appropriate intervals. 
This means that our algorithmic results applies, for example, 
to parallel pancake distributions, as long as the thickness of the pancakes 
is no more than $O(\eps/\sqrt{\log(d)})$.}

Conditions~\ref{anticoncentrationCondition} and \ref{concentrationCondition}
are technical conditions that are needed for our particular algorithm to work.
However, note that if Condition~\ref{anticoncentrationCondition} is not satisfied,
then it is reasonably likely that $O(d)$ random samples from $\p_v^A$ will have
much smaller variance in the $v$-direction than in any of the orthogonal directions.
This provides a much easier method for approximating $v$.
Condition \ref{concentrationCondition} is essentially required to guarantee that we do not
need to deal with unlimited precision to approximate points.
However, it is easy to see that if this condition is violated, one can approximate $v$ simply
by normalizing any samples from $\p_v^A$ with $\ell_2$-norm more than $d$.

We prove the following theorem:

\begin{theorem}[Main Result] \label{thm:main-ngca-disc}
Under Assumption~\ref{cond:A-disc}, if $\eps < 2^{-\Omega(d k^2)} B^{-\Omega(k)}$
with sufficiently large implied \new{universal} constants \new{in the big-$\Omega$},
there exists an algorithm that draws $m = 2d$ i.i.d.\ samples from $\p_v^{A}$
for an unknown unit vector $v \in \R^d$, runs in time $\poly(d, \new{k, \log B})$,
and outputs a vector $v^{\ast}$ such that with constant probability
either $\|v^\ast - v\|_2$ is small or $\|v^\ast + v\|_2$ is small.
\end{theorem}

We note that Theorem~\ref{thm:main-ngca-disc} only guarantees an approximation of either $v$ or $-v$.
Such a guarantee may be inherent, as if $A$ is a symmetric distribution we have that $\p_v^A = \p_{-v}^A$.

\subsection{Overview of Techniques} \label{ssec:techniques}

We begin by considering the simple case where the univariate distribution $A$ is supported {\em exactly} on integers.
This special case provides a somewhat simpler version of our algorithm while capturing some of the key ideas.
In this case, we draw $m=d+1$ i.i.d.\ samples \new{$x_i \in \R^d$} from $\p^A_v$ and note that (with high probability)
they will satisfy a unique (up to scaling) linear relation $\sum_{i=1}^m c_i x_i = 0$,
\new{for some $c_i \in \R$ with at least one $c_i \neq 0$.}
In particular, we have that $\sum_{i=1}^m c_i (v \cdot x_i) = 0$. Since the quantities $v \cdot x_i$ for $i \in [m]$ are all integers,
we hope to solve for them by finding the (with high probability unique, up to scaling)
integer linear relation among the $c_i$'s. \new{It turns out that this can be achieved by leveraging
the Lenstra-Lenstra-Lovasz (LLL) lattice basis reduction algorithm.}
Having found an integer solution $\sum_{i=1}^m c_i n_i = 0$ for $n_i \in \Z$,
we can solve the system of linear equations
$v\cdot x_i = n_i$, $i \in [m]$, for \new{the hidden vector} $v$.

We now proceed to deal with the case where $A$ is no longer supported on integers,
but is instead supported on elements that are {\em close} (within some additive error $\eps$) to integers.
In this case, we will similarly have $\sum_{i=1}^m c_i (v\cdot x_i) = 0$, which means that if $n_i$
is the integer \new{closest} to $v \cdot x_i$, we will have that $\sum_{i=1}^m c_i n_i$ is close to $0$.
In order to solve for this \new{near-integer linear-relation}, we make essential use of basic lattice techniques.
In particular, \new{for $n = (n_i)_{i=1}^m \in \Z^m$,} we define the quadratic form
$Q(\new{n}) := \sum_{i=1}^m n_i^2 + N \left( \sum_{i=1}^m c_i n_i \right)^2$,
for some appropriately large $N$. Note that integer vectors with small norm under $Q$
must have $|n_i|$ small for all $i \in [m]$ and have $\left| \sum_{i=1}^m c_i n_i \right|$ be very small.
It is not hard to show that if $\eps$ is sufficiently small and $N$ is chosen appropriately,
with high probability over the samples $x_i$, taking $n_i$ to be the integer closest to $v\cdot x_i$
for each $i \in [m]$ will give substantially the smallest non-zero norm under $Q$.
Therefore, using the LLL algorithm to find an approximate smallest vector will return (some multiple of) this vector.
Given the $n_i$'s, we note that $v\cdot x_i \approx n_i$ for all $i \in [m]$, and we can then
use least-squares regression to solve for an approximation to $v$.

Unfortunately, if the above approach is applied naively, it will work only if $\eps$ is assumed to be exponentially small in $d^2$,
\new{i.e., $2^{-\Omega(d^2)}$,} rather than in $d$.
This is because the LLL algorithm only guarantees a $2^{O(d)}$-approximation to the smallest vector.
The $\sum_{i=1}^m n_i^2$-term in $Q$ ensures that any such vector will have $n_i$ at most exponential \new{in $d$}.
But given that there are $2^{\Omega(d^2)}$ integer vectors with coefficients of this size,
we can expect one to randomly have $\left| \sum_{i=1}^m c_i n_i \right|$ be only $2^{-\Omega(d^2)}$.
This will be distinguishable from the vector we are looking for only if $\eps$ is smaller than this quantity.

In order to fix this issue, instead of taking only $m = d+1$ samples from $\p_v^A$,
we instead draw $m = 2d$ samples. These now have $d$ linear relations and we note that the vector of $n_i$'s
should approximately satisfy all of them. In particular, letting $V$ be the vector space of linear relations satisfied by the $x_i$'s,
we consider the quadratic form defined by $Q(n): = \|n\|_2^2 + N \left\| \mathrm{Proj}_V(n) \right\|_2^2$.
This improves things because it is now much less likely that one of our
$2^{O(d^2)}$ ``small'' $n$'s will randomly have a small projection onto $V$.
This allows us to operate even when $\eps$ is only \new{$2^{-\Omega(d)}$}.
Note that we cannot hope to do much better than this because the LLL algorithm
will still have an exponential gap between the shortest vector and the one that it finds.

Finally, we are also able to extend our algorithm to the setting where the distribution $A$
is not supported on integers, but instead on numbers of the form $\sum_{j=1}^k a_j r_j$,
where the $a_j$'s are (not too large) integers and the $r_j$'s are some $k$ specific (known) real numbers.
In this more general setting, instead of $v\cdot x_i \approx n_i$, we will have that
$v\cdot x_i \approx \sum_{j=1}^k n_{i, j} r_j$, for some integers $n_{i,j}$, \new{$i \in [m], j \in [k]$}.
We then set-up a quadratic form similar to the one before, namely
$Q(n) = \|n\|_2^2 + N \left \| \mathrm{Proj}_V (t) \right \|_2^2$,
where $t = (t_i)_{\new{i \in [m]}}$ is the vector with coordinates $t_i = \sum_{j=1}^k n_{ij} r_j$
for some integers $n_{ij}$. Once again, the correct integer vector $n$
will be an unusually small vector with respect to this quadratic form; and if we can find it,
we will be able to use it to approximate the hidden direction $v$.

\new{A subtle issue in this case is that}
the correct vector (and multiples) need not be the {\em only} small vectors in this lattice.
In particular, if the $r_j$'s satisfy an approximate linear relation $\sum_{j=1}^k k_j r_j \approx 0$,
then letting $n_{ij} = k_j \cdot \delta_{i, i_0}$, for some $i_0$, will also have $\mathrm{Proj}_V(t)$ small,
because $t$ will be small. To deal with this issue, we will need to apply the LLL algorithm
and take not just the single smallest vector, but the smallest few vectors (in a carefully selected way).
We can then show that the true vector $n$ that we are looking for is in the subspace spanned by these vectors.
By finding a lattice vector in this space such that $t$ is large but $\mathrm{Proj}_V(t)$ is small,
we can find a $t$ where each $t_i$ is approximately some multiple of $v\cdot x_i$ for all $i$.
Using this $t$, we can solve for $v$ \new{as before}.

\paragraph{Independent Work.} 
Concurrent and independent work by~\cite{Zadik21LLL} obtained a similar
algorithm for NGCA under similar assumptions, 
by also leveraging the LLL algorthm. \new{More concretely, the algorithm of~\cite{Zadik21LLL}
efficiently solves the NGCA problem when $A$ is a discrete distribution on an integer lattice,
roughly corresponding to the $k=1$ and $\eps=0$ case of our result.}

\section{Proof of Theorem~\ref{thm:main-ngca-disc}}

\medskip

The pseudo-code for our algorithm is given below.

\bigskip

\fbox{\parbox{6.1in}{
{\bf Algorithm} {\tt LLL-based-NGCA}\\
Input:  $m = 2d$ i.i.d.\ samples \
from $\p_v^A$, where $A$ satisfies Assumption~\ref{cond:A-disc} \new{for given real numbers $r_1,r_2,\ldots,r_k$ and some} parameters $k, B, \eps$ with $\eps < 2^{-C' d k^2} B^{-C' k}$,
where $C'>0$ is a sufficiently large constant.  \\
\vspace{-0.5cm}
\begin{enumerate}

\item Let $N = 2^{Cmk^2} B^{Ck}$ be a positive integer, for $C$ a sufficiently large constant,
such that $N < 1/\eps^2$.

\item Let $x_1,x_2,\ldots,x_m$ be $m$ i.i.d. samples from $\p_v^A$ each rounded to the nearest multiple of $\delta = \eps/N^2$.

\item \label{kernelStep} Let $S$ be the $d\times m$ matrix with columns $x_1,x_2,\ldots,x_m$
and $V$ be the right kernel of $S$.

\item \label{latticeStep} Define the quadratic form $Q$ on $\Z^{m\times k}$
such that for an input vector $n = \{n_{i,j}\}_{i \in [m], j \in [k]}$ we have that
$
Q(n) := \sum_{i=1}^m \sum_{j=1}^k n_{i,j}^2 +
N \left\| \mathrm{Proj}_V\left(\left\{\sum_{j=1}^k n_{i,j} r_j \right\}_{i \in [m]} \right) \right \|_2^2 \;.
$
\item \label{LLLStep} Compute a $\delta$-LLL reduced basis, for $\delta = 3/4$, $\{ b_1,b_2,\ldots,b_{mk}\}$ for $Q$,
\new{where $b_i \in \R^{m \times k}$.}

\item \label{GSStep} Apply the Gram-Schmidt \new{orthogonalization} process to the $b_i$'s,
using $Q$ as our norm, \new{to obtain an orthogonal basis $\{b^{\ast}_1,b^{\ast}_2, \ldots, b^{\ast}_{mk}\}$}.

\item Let $\ell \new{\in[m k]}$ be the largest integer such that $Q\new{(b_{\ell}^{\ast})} \leq mkB^2+Nmk\eps^2.$
Let $W$ be the real span of $\{ b_1,b_2,\ldots,b_{\ell} \}$.

\item \label{eigenvalueStep} Consider the quadratic form $R$ on $\R^{m\times k}$ defined by
$R(n) = \sum_{i=1}^m \left( \sum_{j=1}^k n_{i,j}r_j \right)^2$.
For a sufficiently large universal constant $C>0$,
find a vector $w = \new{\{w_{i,j}\}_{i \in [m], j \in [k]}} \in W$ with
$Q(w) = 2^{Cmk} \, B^2$ such that $R(w)/Q(w)$ is approximately maximized.
Note that this can be done with an eigenvalue computation.

\item Write the vector $w$ in the form  $w =  \sum_{i=1}^\ell c_i b_i$, for some $c_i \in \R$.
Let $w'  \new{=}  \sum_{i=1}^\ell c'_i b_i$, where $c_i'$ is the nearest integer to $c_i$.

\item \label{LeastSquaresStep} Let $v^{\ast}$ be the minimizer of $\sum_{i=1}^m \left( v^\ast \cdot x_i - \sum_{j=1}^k w'_{i,j} r_j \right)^2.$
Note that this can be found using least squares regression.
Return the normalization of $v^\ast$.
\end{enumerate}
}}

\bigskip

To begin the analysis, we first analyze the infinite precision version of this problem,
ignoring the rounding and instead simply letting the $x_i$'s be i.i.d.\ samples from $\p_v^A$.

We begin by analyzing some of the basic properties of the above procedure.
We start by showing that with high probability over our set of samples \new{$x_1, \ldots, x_m$}
\new{the quadratic form} $Q(n)$ is small if and only if
the vector $y^{\new{(n)}} \new{\in \R^m}$ with coordinates $\sum_{j=1}^k n_{i,j} r_j$, \new{$i \in [m]$,}
is approximately a multiple of the vector $y \new{\in \R^m}$ with coordinates $v\cdot x_i$, \new{$i \in [m]$}.

Specifically, we prove the following lemma:

\begin{lemma}\label{QBoundsLem}
Let $y\in \R^m$ be the vector with coordinates $y_i := v\cdot x_i$, $i \in [m]$.
Consider a \new{vector} $\new{n =} (n_{i, j}) \in \R^{m \times k}$.
Let $y^{\new{(n)}} \new{\in \R^m}$ be the vector with \new{coordinates $(y^{\new{(n)}})_i: = \sum_{j=1}^k n_{i,j} r_j$,
$i \in [m]$, and  let $(y^{\new{(n)}})' \new{\in \R^m}$ be the component of $y^{\new{(n)}}$ orthogonal to $y$.}
Then we have that
\begin{equation}\label{QUpperBoundEqn}
Q(n) \leq \|n\|_2^2 + N \, \|(y^{\new{(n)}})' \|_2^2 \;.
\end{equation}
Furthermore, for any positive integer $M$ and with high probability over the choice of $S = \new{[x_i]_{i=1}^m}$,
for all such $n$ with $\|n\|_2 \leq M \, m \, k$, we have that:
\begin{equation}\label{QLowerBoundEqn}
Q(n) \geq N \, O(M)^{-mk/(m-d)} \, \|(y^{\new{(n)}})'\|_2^2 \;.
\end{equation}
\end{lemma}
\begin{proof}
We start by writing $S$ using an orthonormal basis for $\R^d$
in which $v$ is the first vector.
(Note that changing the basis we use for $\R^d$ does not change $V$.)
In this basis, observe that $y$ is the first row of $S$.
Moreover, all other entries of $S$ in this basis are independent standard Gaussians.
Thus, we can write $S$ as $\left[\begin{matrix} y \\ G\end{matrix} \right]$,
where $G$ is an independent Gaussian matrix. Note that the kernel of $G$ is a random subspace of $\R^m$
of dimension $m-d+1$. Thus, after conditioning on $y$,
$V$ is a random $(m-d)$-dimensional subspace orthogonal to $y$.
Also note that we can generate a subspace with the same distribution
by taking the span of $m-d$ independent standard orthogonal-to-$y$ Gaussians.

We next consider a vector $\new{n} = (n_{i,j}) \in \R^{m\times k}$.
Let $y^{\new{(n)}} \in \R^m$ be the vector with coordinates $\sum_{j=1}^k n_{i,j} r_j$,
and let $(y^{\new{(n)}})' \new{\in \R^m}$ be the component of $y^{\new{(n)}} $ orthogonal to $y$.
To begin with, we note that $\|\mathrm{Proj}_V(y^{\new{(n)}})\|_2 \leq \|(y^{\new{(n)}})' \|_2$,
and this implies Equation \eqref{QUpperBoundEqn}.

We prove Equation \eqref{QLowerBoundEqn} by a union bound over the $O(M)^{mk}$
many vectors of appropriate norm. In particular, fix such an $n$.
Recall that conditioning on $y$, the kernel of $v$ is the span of $g_1,g_2,\ldots,g_{m-d}$,
where the $g_i$ are independent orthogonal-to-$y$ Gaussians.
Note that $\|\mathrm{Proj}_V(y^{\new{(n)}} )\|_2 \geq \max_i |g_i \cdot y^{\new{(n)}} |/\|g_i \|_2.$
Note that $g_i \cdot y^{\new{(n)}}$ is distributed like a Gaussian with standard deviation $\|(y^{\new{(n)}})' \|_2$.
For $\delta > 0$, it is not hard to see that for each $i$ we have that
$|g_i \cdot y^{\new{(n)}}|/\|g_i \|_2 < \delta/\sqrt{mk}$ with probability $O(\delta)$
(for example because the probability that $|g_i \cdot y^{\new{(n)}}| < t\delta$
and $\|g_i\|_2 > (t-1)\sqrt{mk}$ is $O(\delta/t^2)$ for any positive integer $t$).
Thus, the probability that $\|\mathrm{Proj}_V(y^{\new{(n)}})\|_2 < \delta/\sqrt{mk}$
is at most $O(\delta)^{m-d}$. Letting $\delta$ be equal to $(CM)^{-mk/(m-d)}$,
for a sufficiently large constant $C$, yields the result.
\end{proof}

We will henceforth assume that \new{the high probability conclusion of}
Lemma \ref{QBoundsLem} holds for the samples our algorithm has selected
with $M := 2^{2 C m k} B$, for $C>0$ a sufficiently large universal constant.
Given this \new{assumption}, we next need to analyze which vectors give us small values of $Q$
and what this means about the output of our call to the LLL algorithm.
In particular, there is a particular vector $n^{\ast}$ that would cause $t$ to approximate $y$.
We claim that $Q(n^{\ast})$ is small and that this in turn implies that $n^{\ast}$
is an integer linear combination of $b_1,b_2,\ldots,b_\ell$.

\new{By assumption}, each $y_i = v\cdot x_i$ is within additive $\eps$ of $\sum_{j=1}^k n^{\ast}_{i,j} r_j$,
\new{for some $n^{\ast}_{i, j} \in \Z$}.
Combining these $n^{\ast}_{i,j}$'s, we get a single vector $n^{\ast} \new{= (n^{\ast}_{ij})_{i \in [m], j \in [k]} \in \Z^{m \times k}}$
which has all entries with absolute value at most $B$,
and by Lemma \ref{QBoundsLem} satisfies $Q(n^{\ast}) \leq m k B^2 + N m k \eps^2$.
Note that $n^{\ast}$ is a linear combination of the $b_i$'s, namely
$n^{\ast} = \sum_{i=1}^{m k} c_i b_i$.
Let $t$ be the largest $i$ such that $c_i \neq 0$.
Note that we can also write $n^{\ast}$ as $\sum_{i=1}^{m k} c_i' b_i^{\ast}$, for some real $c_i'$,
and that $c_t' = c_t$. Since the $b_i^{\ast}$ are orthogonal with respect to the quadratic form $Q$,
this implies that
$$Q(n^{\ast}) \geq Q(c_t' b_t^{\ast}) \geq Q(b_t^{\ast}) \;.$$
In particular, this means that $Q(b_t^{\ast}) \leq m k B^2+ N m k \eps^2$. By our choice of $\ell$,
this implies that $t \leq \ell$, and in particular that $n^{\ast}$ is a linear combination of
$b_1,b_2,\ldots,b_\ell$.

Unfortunately, we cannot necessarily find $n^{\ast}$ within this subspace.
However, it will suffice for our purposes to find a vector $z$ for which $\|z\|_2$ is large,
but the part of $z$ orthogonal to $y$ is small. To do this, it will suffice to find an integer vector $n$
for which $R(n)$ is large (implying that $\|y^{\new{(n)}}\|_2$ is large), but for which $Q(n)$ is small
(implying that $y^{\new{(n)}}$ is nearly orthogonal to $y$).
We know that $n^{\ast}$ is such a vector and that it is somewhere in $W$.
It now remains to find it.

Note that $n^{\ast} \in W$.
Note that $R(n^{\ast}) = \|y^{\new{(n^\ast}}) \|_2^2 \geq \|y\|_2^2/2 + O(m\eps^2).$
By the anti-concentration Condition \ref{anticoncentrationCondition},
with constant probability over the choice of $y$, this is $\Omega(1/d^2)$.
On the other hand, we have that $Q(n^{\ast}) \leq m k (B^2+N \eps^2)$.
Therefore, we have that
$$R(n^{\ast})/Q(n^{\ast}) \geq \Omega(1/(d^2 mk B^2)) \;.$$
Given our algorithm's choice of $w$, we have that
$R(w)/Q(w) \geq \Omega(1/(d^2 mk B^2)).$
On the other hand, we note that for $i\leq \ell$ we have that
$$Q(b_i) \leq 2^{mk} Q(b_{\ell}^{\ast}) \leq 2^{m k}(mkB^2) \;.$$
This in particular follows from the fact that $Q(b_{i+1}^{\ast}) \geq Q(b_i^\ast)/2$ for positive integers $i$.
The latter statement can be derived, for example, from p.\ 86 of~\cite{Cohen-nt-book}).
This means that $\|b_i \|_2^2 \leq 2^{m k}(m k B^2)$,
and thus $R(b_i) \leq 2^{O(m k)} B^2$. This implies that
$R(w-w') \leq 2^{O(m k)}(B^2+N\eps^2).$
However, since $Q(b_i) \leq 2^{mk}(mkB^2)$,
by similar reasoning, we obtain that $Q(w-w') \leq 2^{O(mk)} B^2$.
Together, this implies that $\|w'\|_2^2 \leq Q(w') = \Theta(2^{C m k} B^2)$,
and that
$$R(w')/Q(w') = \Omega(1/(d^2 mk B^2)) \;.$$
Assuming the high probability statement of Lemma \ref{QBoundsLem}
with $M := 2^{2Cmk}B$, we have that
$$Q(w') \geq N 2^{-O(m^2k^2/(m-d))}B^{-mk/(m-d)}\|\new{(y^{(w')})'}\|_2^2 \;.$$
This implies that $\|\new{(y^{(w')})'} \|_2^2 \leq N^{-1} 2^{O(m^2k^2/(m-d))}B^{O(mk/(m-d))}.$
Note that this means that the vector with coordinates $\sum_{j=1}^k w'_{i,j} r_j$
is within $N^{-1} 2^{O(m^2k^2/(m-d))}B^{O(mk/(m-d))}$ of some multiple of $y$.
Thus, taking $v^{\ast}$ to be an appropriate multiple of $v$ yields an error of at most
$N^{-1} 2^{O(m^2k^2/(m-d))}B^{O(mk/(m-d))}$ in the defining equation of $v^{\ast}$.

We next need to determine how close the above implies that $v^{\ast}$ will be to a multiple of $v$.
To analyze this, we consider the eigenvalues of the matrix $\sum_{i=1}^m x_ix_i^T$.
By Condition \ref{anticoncentrationCondition}, with large constant probability,
the eigenvalue in the $v$-direction will be at least $\Omega(1/d^2)$.
As the $x_i$'s in orthogonal directions are independent standard Gaussians,
it is not hard to see (for example via a cover argument) that with this large constant probability
all eigenvalues of $\sum_{i=1}^m x_ix_i^T$ are at least $\Omega(1/d^2)$.
The error in \new{the least-squares regression problem} equals
$$
\sum_{i=1}^m \left( v^\ast \cdot x_i - \sum_{j=1}^k w_{i,j}' r_j \right)^2
= \sum_{i=1}^m \left( v^\ast \cdot x_i - \alpha y_i - (\new{(y^{(w')})'})_i \right)^2 \;,
$$
where $\alpha$ is some real multiple.
Notice that $v^\ast \cdot x_i - \alpha y_i = (v^\ast - \alpha v)\cdot x_i$.
Therefore the above is
$$
(v^\ast - \alpha v)^T \sum_{i=1}^m x_i x_i^T (v^\ast - \alpha v) +
 O\left( (m/d) \|\new{(y^{(w')})'} \|_2^2 + (m/d) \, \|\new{(y^{(w')})'}\|_2  \, \|v^\ast - \alpha v\|_2 \right) \;.
$$
In particular, noting that setting $v^\ast = \alpha v$ obtains a value of
$O(m/d) \left( N^{-1} 2^{O(m^2k^2/(m-d))}B^{O(mk/(m-d))} \right)^2$,
the true $v^{\ast}$ must satisfy
$$\|v^\ast - \alpha v\|_2 \leq N^{-1} 2^{O(m^2k^2/(m-d))}B^{O(mk/(m-d))} \;.$$
On the other hand, since $R(w') > 1$, we have that $\|\new{(y^{(w')})'} \|_2^2 > 1$,
which (assuming that all $x_i$'s have norm $O(\sqrt{d})$ \new{which holds with high probability})
implies that $\|v^{\ast} \|_2 \gg 1/\sqrt{d}$. This means by the above that the normalization of $v^{\ast}$
is within \new{$\ell_2$-error}
$$N^{-1} 2^{O(m^2k^2/(m-d))}B^{O(mk/(m-d))}$$
of $\pm v$.

Since we have selected $m=2d$, $N = 2^{Cmk^2} B^{Ck}$, for $C$ a sufficiently large constant,
and $\eps < 2^{-C' d k^2} B^{-C' k}$, for some sufficiently large constant $C'$,
it follows that the normalization of $v^{\ast}$ 
is exponentially close to $\pm \new{v}$.

Next we need to show that rounding the $x_i$'s does not affect the correctness of our procedure.
For this, we note that the above analysis only needed the following facts about the $x_i$:
\begin{enumerate}
\item Lemma \ref{QBoundsLem} holds for $M=2^{2Cmk}B$.
\item \label{nonsingularCondition} $\sum_{i=1}^m x_i x_i^T \succeq \Omega(I/d^2)$.
\item \label{closeToLatticeCondition} $v\cdot x_i$ is within $2^{-\Omega( d k^2)} B^{-\Omega( k)}$ 
(with sufficiently large constants in the big-$\Omega$) of some integer linear combination of the $r_i$'s
with coefficients of absolute value at most $B$ for all $i$.
\end{enumerate}
We note that these hold with reasonable probability by the above.
We claim that if they hold for the unrounded $x_i$'s and if the $x_i$'s have absolute value at most $\poly(d)$
(which happens with constant probability by Condition \ref{concentrationCondition}),
then they hold for the rounded $x_i$'s,
perhaps with slightly worse implied constants in the big-$O$ and big-$\Omega$ terms.

To show this, we begin with Condition \ref{closeToLatticeCondition}.
This still holds since rounding an $x_i$ changes the value of $v\cdot x_i$ by at most $d\delta < \eps$.

For Condition \ref{nonsingularCondition}, we note that changing each coordinate of $x_i$ by $\delta$
changes $\sum_{i=1}^m x_ix_i^T$ by at most $d m\delta \max_i \|x_i\|_2$ in Frobenius norm. 
As this is much less than $1/d^2$, the minimum eigenvector of $\sum_{i=1}^m x_ix_i^T$ is still large enough after the rounding.

Finally, for Lemma \ref{QBoundsLem}, we note that the argument for Equation~\eqref{QUpperBoundEqn} still applies.
For Equation \eqref{QLowerBoundEqn}, we note that for a vector $z$, $\mathrm{Proj}_V(z) = z-t$,
where $t$ is the unique vector in the range of $S^T$ such that $St = Sz$.
From this, we conclude that $t = S^T (S S^T)^{-1} S z$.
We claim that rounding this does not change the value of $t$
(or, therefore, the value of $\|t-z\|_2^2$) by much.
In particular, it is easy to see that rounding changes $S$ and $S^T$ by $O(md \delta)$ in Frobenius norm.
The effect on $(S S^T)^{-1}$ is more complicated;
but we know that $SS^T = \sum_i x_i x_i^T \succeq \Omega((1/d^2)) \ I$.
This and the fact that the rounding changes $SS^T$ by relatively little in terms of Frobenius norm,
suffices to imply that the rounding does not change much the value of $Q(n)$,
for any vector $n$ with coefficients of absolute value $M$.

Having established correctness, we need to bound the runtime.
This is relatively straightforward, as we have to solve problems in dimension $\poly(md)$ with $\poly(md\log(B))$ bits of precision.
In particular, Step~\ref{kernelStep} boils down to row-reduction; Step~\ref{latticeStep} requires computing a projection matrix;
Step~\ref{LLLStep} uses the LLL algorithm; Step~\ref{eigenvalueStep} can be done via an approximate eigenvalue computation;
and Step~\ref{LeastSquaresStep} is least squares. Each of these operations can be performed in time that is polynomial
in the dimension of the problem and in the number of bits of precision required.

This completes the proof of Theorem~\ref{thm:main-ngca-disc}. \qed

\begin{remark}
{\em We remark that our algorithm works with any number $m>d$ samples, 
as long as $\eps$ is less than $2^{-\Omega( d k^2 m / (m-d))} B^{-\Omega(k m/(m-d))}$ 
for sufficiently large constants in the big-$\Omega$'s. 
For example, one could take $m=d+1$ samples, as long as $\eps < 2^{-C'd^2k^2}B^{-C' dk}$.}
\end{remark}

\paragraph{Acknowledgements.}
We thanks Sam Hopkins and Aaron Potechin for useful discussions about the low-degree
conjecture and Sum-of-Squares lower bounds.

\bibliographystyle{alpha}
\bibliography{allrefs}

\end{document}